\documentclass{sig-alternate}
\usepackage{mathrsfs}
\usepackage{amsmath}
\usepackage{amssymb}
\usepackage{color}
\usepackage[colorlinks,linkcolor=black,anchorcolor=black,citecolor=black,urlcolor=black]{hyperref}

\newcommand{\bA}{ {\mathbb A}}
\newcommand{\bE}{ {\mathbb E}}
\newcommand{\bB}{ {\mathbb B}}

\newcommand{\bF}{ {\mathbb F}}
\newcommand{\cE}{ {\cal  E}}
\newcommand{\bZ}{ {\mathbb Z}}
\newcommand{\cP}{ {\cal P}}
\newcommand{\cG}{ {\cal G}}

\newcommand{\cQ} { {\cal Q}}

\newcommand{\vt} { {\bf t}}
\newcommand{\vx} { {\bf x}}
\newcommand{\vy} { {\bf y}}

\newcommand{\pa } {  {\partial} }
\newcommand{\den } {  {\rm den} }
\newcommand{\de } {\delta}
\newcommand{\num }{ {\rm num}}

\newcommand{\si}{ \sigma }

\newcommand{\bone} { {\boldsymbol 1}}
\newcommand{\bzero} { {\boldsymbol 0}}

\newcommand{\lde} { {\ell \delta}}
\newcommand{\lsi} { {\ell \si}}
\newcommand{\ltau} { {\ell \tau}}

\newcommand{\bxi}{{\boldsymbol \xi}}

\newtheorem{theorem}{Theorem}[section]
\newtheorem{prop}[theorem]{Proposition}
\newtheorem{cor}[theorem]{Corollary}
\newtheorem{lemma}[theorem]{Lemma}

\newtheorem{example}{Example}[section]
\newtheorem{remark}[example]{Remark}

\newtheorem{fact}{Fact}[section]

\begin{document}
\conferenceinfo{ISSAC XXXX} {}
\CopyrightYear{2011}
\crdata{}
\clubpenalty=10000
\widowpenalty = 10000

\title{On the Structure of Compatible Rational Functions\thanks{This work was supported in part
by two grants of NSFC No.\ 60821002/F02 and No.\ 10901156. The first author was a PhD student in the Chinese Academy of
Sciences and INRIA, Paris-Rocquencourt when the first draft of this paper was written.
He is now a post doctoral fellow at RISC-Linz, and
acknowledges the financial
support by Austrian FWF grant Y464-N18}}

\numberofauthors{1}
\author{
Shaoshi Chen$^{1,2,3}$, Ruyong Feng$^1$, Guofeng Fu$^1$, Ziming Li$^1$ \\
       \affaddr{$^1$Key Lab of Math.-Mech.\, AMSS, \,Chinese Academy of Sciences, Beijing 100190, (China)}\\
       \affaddr{$^2$ Algorithms Project-Team, INRIA, Paris-Rocquencourt, 78513 Le Chesnay, (France)}\\
       \affaddr{$^3$ RISC, Johannes Kepler University, 4040 Linz, (Austria)}\\
      \email{$\{$schen, ryfeng$\}$@amss.ac.cn, \, $\{$fuguofeng, zmli$\}$@mmrc.iss.ac.cn}
}

\maketitle
\begin{abstract}
A finite number of rational functions are compatible if they satisfy
the compatibility conditions of a first-order linear functional
system involving differential, shift and $q$-shift operators. We
present a theorem that describes the structure of compatible
rational functions. The theorem enables us to decompose a solution
of such a system as a product of a rational function, several
symbolic powers, a hyperexponential function, a hypergeometric term,
and a $q$-hypergeometric term. We outline an algorithm for computing
this product, and present an application.
\end{abstract}
\category{I.1.2}{Computing Methodologies}{Symbolic and Algebraic Manipulation}[Algebraic Algorithms]
\terms{Algorithms, Theory}
\keywords{Compatibility conditions, compatible rational functions,
hyperexponential function, ($q$-)hypergeometric term}

\section{Introduction} \label{SECT:intro}
A linear functional system consists of linear
partial differential, shift and $q$-shift operators.
The commutativity of these operators implies that
the coefficients of a linear functional system satisfy
compatibility conditions.

A nonzero solution of a first-order linear partial differential system in one unknown function is
called a hyperexponential function. Christopher and
Zoladek~\cite{Christopher1999,Zoladek1998} use the compatibility
(integrability) conditions to show that a hyperexponential function
can be written as a product of a rational function, finitely many
power functions, and an exponential function. Their results generalize
a well-known fact, namely, for a rational function~$r(t)$,
$$\exp\left(\int r(t) dt \right) =  f(t) r_1(t)^{e_1} \cdots r_m(t)^{e_m} \exp(g(t)),$$
where~$e_1, \ldots, e_m$ are constants, and~$f, r_1, \ldots, r_m, g$ are rational functions.
The generalization is
useful to compute Liouvilian first integrals.

A nonzero solution of a first-order
linear partial  difference system in one unknown term is called a hypergeometric term.
The Ore-Sato  Theorem~\cite{Ore1930,Sato1990} states that
a hypergeometric term is a product of a rational function, several power functions
and factorial terms. A $q$-analogue of the Ore-Sato theorem is given in~\cite{Gelfand,Chen2005}.
All these results are based on compatibility conditions. The Ore-Sato theorem was rediscovered
in one way or another, and is important for the proofs
of a conjecture of Wilf and Zeilberger about
holonomic hypergeometric terms~\cite{AbramovPetkovsek2001, AbramovPetkovsek2002a, Payne1997}.
This theorem and its $q$-analogue also play a crucial role
in deriving criteria on the existence of telescopers for hypergeometric and
$q$-hypergeometric terms, respectively~\cite{Abramov2003, Chen2005}.

Consider a first-order mixed system
$$ \left\{ \frac{\pa z(t, x)}{\pa t} = u(t, x) z(t, x), \,\,
           z(t, x+1) = v(t, x) z(t, x) \right\},$$
where~$u$ and~$v$ are rational functions with~$v \neq 0$.
Its compatibility condition is~$\pa v(t, x) / \pa t = v(t,x) (u(t, x+1) - u(t, x)).$
By Proposition~5 in~\cite{Feng2010b},
a nonzero solution of the above system can be written as a product~$f(t, x) r(t)^x \cE(t) \cG(x)$,
where~$f$ is a bivariate rational function in~$t$ and~$x$,
$r$ is a univariate rational function in~$t$, $\cE$ is a
hyperexponential function in~$t$, and~$\cG$ is a hypergeometric
term in~$x$.
This proposition is used to compute Liouvillian
solutions of difference-differential systems.

In fact, the above proposition is also fundamental for
the criteria on
the existence of telescopers when both differential and shift operators are involved~\cite{CCFL2010}.
This motivates us to generalize the proposition to include differential, difference and $q$-difference cases.
Such a generalization will enable us to establish the existence of telescopers when
both differential (shift) and $q$-shift operators appear.
Next, the proof of the Wilf-Zeilberger conjecture for hypergeometric terms is based on
the Ore-Sato theorem. So
it is reasonable to expect that
a structural theorem on compatible rational functions with respect to
differential, shift and $q$-shift operators helps us study
the conjecture in more general cases.

The main result of this paper is Theorem~\ref{TH:dsq}
which reveals a special structure of compatible rational functions.
By the theorem, a hyperexponential-hypergeometric solution, defined in Section~\ref{SECT:as},
is a product of a rational function, several symbolic powers,
a hyperexponential function, a hypergeometric term, and a $q$-hypergeometric term (see Proposition~\ref{PROP:hyper}).
This paves the way to decompose such solutions by
Christopher-Zoladek's generalization, the Ore-Sato Theorem, and its $q$-analogue.

This paper is organized as follows. The notion of compatible
rational functions is introduced in Section~\ref{SECT:as}.
The bivariate case is studied in
Section~\ref{SECT:bc}. After presenting a few preparation lemmas in Section~\ref{SECT:mdf}, we prove in Section~\ref{SECT:st}
a theorem that describes the structure of compatible rational functions.
Section~\ref{SECT:ap} is about algorithms and applications.
\section{Compatible rational functions} \label{SECT:as}
In the rest of this paper, $\bF$ is a field of characteristic zero.
Let~$\vt =(t_1, \ldots, t_l)$, $\vx=(x_1, \ldots, x_m)$
and~$\vy=(y_1, \ldots, y_n)$. Assume that~$q_1, \ldots, q_n \in \bF$
are neither zero nor roots of unity. For an element~$f$ of~$\bF(\vt, \vx, \vy)$,
define~$\delta_i(f) =\pa f/\pa t_i$
for all~$i$ with~$1 \le i \le l$,
\[ \si_j(f(\vt, \vx, \vy)) = f(\vt, x_1, \ldots, x_{j-1}, \, x_j+1, \, x_{j+1}, \ldots, x_m, \vy) \]
for all~$j$ with~$1 \le j \le m$, and
\[ \tau_k(f(\vt, \vx, \vy)) = f(\vt, \vx, y_1, \ldots, y_{k-1}, \, q_k y_k, \, y_{k+1}, \ldots, y_n) \]
for all~$k$ with~$1 \le k \le n$. They are called derivations, shift operators, and $q$-shift
operators, respectively.

Let~$\Delta=\{\delta_1, \ldots, \delta_l, \si_1, \ldots, \si_m, \tau_1, \ldots, \tau_n \}$.
These operators commute pairwise. The field of constants w.r.t.\ an operator in~$\Delta$
consists of all rational functions free of the indeterminate on which the operator acts nontrivially.

By a first-order linear functional system over~$\bF(\vt, \vx, \vy)$,
we mean a system consisting of
\begin{equation} \label{EQ:sys}
\delta_i(z) = u_i z, \,\, \si_j(z)=v_j z, \,\, \tau_k(z) = w_k z
\end{equation}
for some rational functions~$u_i, v_j, w_k \in \bF(\vt, \vx, \vy)$ and for all~$i, j, k$
with~$1 \le i \le l$, $1 \le j \le m$ and~$1 \le k \le n$. System~\eqref{EQ:sys}
is said to be {\em compatible} if
\begin{equation} \label{EQ:neq}
v_1 \cdots v_m w_1 \cdots w_n \neq 0
\end{equation}
and the conditions listed in~\eqref{EQ:dd}-\eqref{EQ:sq} hold:
\begin{equation} \label{EQ:dd}
\delta_i(u_j)=\delta_j(u_i),  \quad  \mbox{$1 \le i < j \le l$,}
\end{equation}
\begin{equation} \label{EQ:ss}
\si_i(v_j)/v_j = \si_j(v_i)/v_i, \quad  \mbox{$1 \le i < j \le m$,}
\end{equation}
\begin{equation} \label{EQ:qq}
\tau_i(w_j)/w_j = \tau_j(w_i)/w_i, \quad  \mbox{$1 \le i < j \le n$,}
\end{equation}
\begin{equation} \label{EQ:ds}
\delta_i(v_j)/v_j = \si_j(u_i) - u_i, \quad \mbox{$1 \le i \le l$ and~$1 \le j \le m$,}
\end{equation}
\begin{equation} \label{EQ:dq}
\delta_i(w_k)/w_k = \tau_k(u_i) - u_i,  \quad  \mbox{$1 \le i \le l$ and~$1 \le k \le n$,}
\end{equation}
\begin{equation} \label{EQ:sq}
\si_j(w_k)/w_k = \tau_k(v_j)/v_j, \quad  \mbox{$1 \le j \le m$ and~$1 \le k \le n$.}
\end{equation}
Compatibility conditions~\eqref{EQ:dd}-\eqref{EQ:sq} are caused by the commutativity of the maps
in~$\Delta$. A sequence of rational functions:~$u_1,$ \, \ldots, \, $u_l,$  \, $v_1,$ \, \ldots, \, $v_m,$  \, $w_1,$ \, \ldots, \, $w_n$
is  said to be {\em compatible} w.r.t.~$\Delta$ if~\eqref{EQ:neq}-\eqref{EQ:sq} hold.

By a $\Delta$-extension of~$\bF(\vt, \vx, \vy)$, we mean a ring extension~$R$ of~$\bF(\vt, \vx, \vy)$
s.t.\ every derivation and automorphism in~$\Delta$ can be extended to a derivation and
a monomorphism from~$R$ to~$R$, and, moreover, the extended maps are commutative with each other.
Given a finite number of first-order compatible systems, one can construct a Picard-Vessiot $\Delta$-extension
of~$\bF(\vt, \vx, \vy)$ that contains \lq\lq all\rq\rq~solutions of these systems. Moreover, every
nonzero solution is invertible.
Details on Picard-Vessiot extensions of compatible systems
may be found in~\cite{Bronstein2005}. More general and powerful extensions are described in~\cite{Hardouin2008}.
By a {\em hyperexponential-hypergeometric solution~$h$ over~$\bF(\vt, \vx, \vy)$}, we mean a nonzero solution of
the system~\eqref{EQ:sys}.
The coefficients~$u_i$, $v_j$ and~$w_k$ in~\eqref{EQ:sys} are called {\em $\delta_i$-, $\si_j$-, and $\tau_k$-certificates} of~$h$,
respectively. For brevity,  we abbreviate  \lq\lq hyperexponential-hypergeometric solution\rq\rq\
as \lq\lq $H$-solution\rq\rq. An $H$-solution is a hyperexponential function when~$m=n=0$ in~$\eqref{EQ:sys}$,
it is a hypergeometric term if~$l=n=0$, and a $q$-hypergeometric term if~$l=m=0$.
\begin{remark} \label{RE:ap}
We opt for the word \lq\lq solution\rq\rq\ rather than \lq\lq function\rq\rq,
since all the~$t_i$, $x_j$ and~$y_k$ are regarded as indeterminates.
It is more sophisticated  to regard hypergeometric
terms as functions of integer variables~\cite{Payne1997, AbramovPetkovsek2002a,AbramovPetkovsek2008}.
\end{remark}

As a matter of notation, for an element~$f \in \bF(\vt, \vx, \vy)$,
the denominator and numerator of~$f$ are denoted~$\den(f)$ and~$\num(f)$,
respectively. Note that $\den(f)$ and~$\num(f)$ are coprime.
For a ring~$\bA$, $\bA^\times$
stands for~$\bA {\setminus} {\{}0{\}}$,  and for a field~$\bE$,~$\overline{\bE}$ stands for the algebraic closure of~$\bE$.
For every~$\phi \in \Delta$ and~$f \in \bF(\vt, \vx, \vy)^\times$, we denote by~$\ell\phi(f)$ the fraction~$\phi(f)/f$.
When~$\phi$ is a derivation~$\delta_i$, $\lde_i(f)$ stands for the logarithmic derivative of~$f$
with respect to~$t_i$. This notation allows us to avoid stacking fractions and subscripts.

Let~$\bE$ be a field and~$t$ an indeterminate. A nonzero element~$f$ of~$\bE(t)$ can be written
uniquely as~$f=p + r$, where~$p \in \bE[t]$ and $r$ is a proper fraction. We say that~$p$ is
the polynomial part of~$f$ w.r.t.~$t$.
\begin{remark} \label{RE:log}
Let~$z {\in} \{ t_1, \, \ldots, \,  t_l, \, x_1, \, \ldots, \, x_m, \,
y_1, \, \ldots, \, y_n \}$ and~$f \in \bF(\vt, \vx, \vy)^\times$.
For all~$i$ with~$1 \le i \le l$,  the polynomial part of~$\lde_i(f)$ w.r.t.~$z$ has degree
at most zero in~$z$.
\end{remark}
\section{Bivariate case} \label{SECT:bc}
In this section, we assume that~$l=m=n=1$.
For brevity, set~$t=t_1$, $x=x_1$, $y=y_1$, $\delta=\delta_1$,
$\si=\si_1$, $\tau{=}\tau_1$, and~$q=q_1$. By~\eqref{EQ:neq},~\eqref{EQ:ds}, \eqref{EQ:dq} and~\eqref{EQ:sq},
three rational functions~$u, v, w$ in~$\bF(t, x, y)$
are $\Delta$-compatible if~$vw \neq 0$,
\begin{equation} \label{EQ:ds1}
\lde(v)=\si(u)-u,
\end{equation}
\begin{equation} \label{EQ:dq1}
\lde(w) = \tau(u)-u,
\end{equation}
\begin{equation} \label{EQ:sq1}
\lsi(w)=\ltau(v).
\end{equation}
Other compatibility conditions become trivial in this case.
\begin{example} \label{EX:mixed}
Let~$\alpha \in \bF(t, y)^\times$. The system consisting of~$\delta(z) = \lde(\alpha) \, x \, z$
and~$\si(z) = \alpha \, z$
is compatible w.r.t.~$\delta$ and~$\si$. Denote a solution
of this system by~$\alpha^x$, which is irrational if~$\alpha \neq 1$.
\end{example}
The next lemma is immediate from~\cite[Proposition 5]{Feng2010b}.
\begin{lemma} \label{LM:fsw}
Let~$u, v \in \bF(t, x, y)$ with~$v \neq 0$.
If~\eqref{EQ:ds1} holds, then~$u = \lde(f) + \lde(\alpha) \, x + \beta$
and~$v {=}  \lsi(f) \,\alpha \lambda$
for some~$f$ in~$\bF(t, x, y)$, $\alpha, \beta$ in~$\bF(t, y),$ and~$\lambda$ in~$\bF(x, y)$.
\end{lemma}
Assume that an $H$-solution~$h$
has $\delta$-certificate~$u$ and $\si$-certificate~$v$.
By Lemma~\ref{LM:fsw},~$h = c f \alpha^x \cE \cG$ in some $\Delta$-ring,
where~$c$ is a constant w.r.t.~$\de$ and~$\si$,~$\cE$ is hyperexponential
with certificate~$\beta$, and~$\cG$ is hypergeometric
with certificate~$\lambda$.

We shall prove two similar results: one is
about differential and $q$-shift variables; the other about shift
and $q$-shift ones.
To this end, we recall some terminologies from~\cite{AbramovPetkovsek2001, AbramovPetkovsek2002a, Hardouin2008}.

Let~$\bA{=}\bF(t,y)$ and~$p \in \bA[x]^\times$. The {\em $\si$-dispersion
of~$p$} is defined to be the largest nonnegative integer~$i$
s.t.~for some~$r$ in~$\overline{\bA}$,~$r$ and~$r+i$ are roots of~$p$.
Let~$f \in \bA(x)^\times$. We say that~$f$ is {\em
$\si$-reduced} if~$\den(f)$ and~$\si^i(\num(f))$ are coprime for
every integer~$i$; and that~$f$ is {\em $\si$-standard} if zero is
the $\si$-dispersion of~$\num(f)\den(f)$. A $\si$-standard rational function
is a $\si$-reduced one, but the converse is false. By Lemma~6.2 in~\cite{Hardouin2008},~$f =
\lsi(a) \, b$ for some~$a, b$ in~$\bA(x)$ with~$b$ being $\si$-standard
or $\si$-reduced.

Let~$\bB=\bF(t,x)$ and~$p \in \bB[y]^\times$. The {\em
$\tau$-dispersion of~$p$} is defined to be the largest nonnegative
integer~$i$ s.t.~for some {\em nonzero}~$r \in \overline{\bB}$, $r$
and~$q^i r$ are roots of~$p$.
In addition, the $\tau$-dispersion of~$p$ is set to be zero if~$p = c y^k$ for some~$c \in \bB$.
Let~$f \in \bB(y)^\times$. The {\em
polar $\tau$-dispersion} is the $\tau$-dispersion
of~$\den(f)$. The notion of $\tau$-reduced and $\tau$-standard
rational functions are defined likewise.
One can write~$f = \ltau(a) \, b$,
where~$a, b \in \bB(y)^\times$ and $b$ is $\tau$-standard
or $\tau$-reduced.

Now, we prove a $q$-analogue of Lemma~\ref{LM:fsw}.
\begin{lemma} \label{LM:dq11}
Let~$u, w {\in} \bF(t, x, y)$ with~$w {\neq} 0$.
If~\eqref{EQ:dq1} holds,
then~$u {=} \lde(f) {+} a$ and~$w {=} \ltau(f) \, b$
for some~$f$ in~$\bF(t, x, y)$,~$a$ in~$\bF(t,x)$, and~$b$ in~$\bF(x, y)$.
\end{lemma}
\begin{proof}
Set~$w = \ltau(f) \, b$  for some~$f, b$ in~$\bF(t, x, y)$
with~$b$ being $\tau$-standard. Set~$b = y^k P/Q$,
where~$P, Q \in \bF(x)[t,  y]$ are coprime, and neither is divisible by~$y$. Since~$b$ is $\tau$-standard,
so is~$P/Q$. Assume~$u = \lde(f) + a$. By~\eqref{EQ:dq1},
\begin{equation} \label{EQ:dq11}
\lde(P/Q) = \tau(a) - a.
\end{equation}
Since~$P/Q$  is $\tau$-standard, the $\tau$-dispersion of~$PQ$ is zero,
and so is the polar $\tau$-dispersion of the left-hand side in~\eqref{EQ:dq11},
which, together with~\cite[Lemma~6.3]{Hardouin2008},
implies that~$a$ belongs to~$\bF(t, x)[y, y^{-1}]$.
Moreover,~$a$ is free of positive powers of~$y$ by Remark~\ref{RE:log} (setting~$z=y$);
and~$a$ is free of negative powers of~$y$,
because neither~$P$ nor~$Q$ is divisible by~$y$. We conclude that~$a$ is in~$\bF(t,x)$.
Consequently,~$\tau(a)=a$. It follows from~\eqref{EQ:dq11}
that~$\delta(P/Q)=0$, i.e.,~$b$ is in~$\bF(x, y)$.
\end{proof}
By the above lemma, an $H$-solution~$h$
can be written as a product of a constant w.r.t.~$\de$ and~$\tau$, a rational function, a hyperexponential
function, and a $q$-hypergeometric term.

The last lemma is a $q$-analogue of~\cite[Theorem 9]{AbramovPetkovsek2002a}.
Our proof is based on an easy consequence of~\cite[Lemma 2.1]{AdPutSinger1997}.
\begin{fact} \label{FA:dn}
Let~$a, b\in \bF(t, x, y)^\times$.
If~$\si(a)=ba$, and~$P$ is an irreducible factor of~$\den(b)$ with~$\deg_x P > 0$,
then~$\si^i(P)$ is a factor of~$\num(b)$ for some nonzero integer~$i$.

The same is true if we swap~$\den(b)$ and~$\num(b)$ in the above assertion.
\end{fact}
\begin{lemma} \label{LM:sq11}
Let~$v, w \in \bF(t, x, y)^\times$. If~\eqref{EQ:sq1} holds,
then $v = \lsi(f) \, a $ and~$w = \ltau(f) \, b$
for some~$f$ in~$\bF(t, x, y)$, $a $ in~$\bF(t, x)$, and~$b$ in~$\bF(t, y)$.
\end{lemma}
\begin{proof} In this proof, $P \mid Q$ means that~$P, Q \in \bF(t)[x,y]^\times$
and~$Q=PR$ for some~$R \in \bF(t)[x,y]$.

Set~$v = \lsi(f) \, a$, where~$f, a \in \bF(t, x, y)$
and~$a $ is $\si$-reduced.
Assume~$w =  \ltau(f) \, b$. By~\eqref{EQ:sq1},
$\lsi(b) = \ltau(a) $, that is,
\begin{equation} \label{EQ:sq11}
\si(b) = g b, \quad
\mbox{where}~g = \frac{\tau(\num(a )) \, \den(a )}{\tau(\den(a )) \, \num(a )}.
\end{equation}
First, we show that~$a $ is the product of an element in~$\bF(t, x)$ and an element in~$\bF(t, y)$.
Suppose the contrary. Then there is an irreducible polynomial~$P \in \bF(t)[x,y]$
with~$\deg_x P {>} 0$ and~$\deg_y P {>} 0$ s.t.~$P$ divides~$\den(a) \num(a )$ in~$\bF(t)[x,y]$.
Assume that~$P \mid \num(a )$. If~$P \nmid \den(g)$, then~$P \mid \tau(\num(a ))$ since~$\num(a )$
and~$\den(a )$ are coprime. So~$\tau^{-1}(P)|\num(a)$.
If~$P| \den(g)$, then~$\si^i(P)\mid \num(g)$ for some integer~$i$ by~\eqref{EQ:sq11} and Fact~\ref{FA:dn}.
Thus,~$\si^i(P) \mid \tau(\num(a ))$, because~$\num(g)$ is a factor of~$\tau(\num(a ))\den(a )$
and~$a $ is $\si$-reduced. This implies~$\si^i \tau^{-1}(P) \mid \num(a )$.
In either case, we have that
$$\si^{j} \tau^{-1}(P) \mid \num(a ) \quad \mbox{for some integer~$j$.}$$
Assume~$P \mid \den(a )$. Then the same argument implies
$$\si^{k} \tau^{-1}(P) \mid \den(a ) \quad \mbox{for some integer~$k$.}$$
Hence, there exists an integer~$m_1$ s.t.~$P_1 := \si^{m_1}\tau^{-1}(P_0)$
is an irreducible factor of~$\den(a ) \num(a )$, where~$P_0=P$. A repeated
use of the above reasoning leads to an infinite sequence of irreducible
polynomials~$P_0, P_1, P_2, \ldots$ in~$\bF(t)[x, y]$ s.t.~$P_i = \si^{m_i} \tau^{-1}(P_{i-1})$ and~$P_i \mid \den(a )\num(a ).$
Therefore, there are two $\bF(t)$-linearly dependent members in the sequence.
Using these two members, we get~$P_0 = c \si^m \tau^n (P_0)$ for some~$c$ in~$\bF(t)$ and~$m, n$ in~$\bZ$ with~$n \neq 0$.
Write
\[ P_0 = p_d(x)y^d + p_{d-1}(x)y^{d-1} + \cdots + p_0(x), \]
where~$d>0$, $p_i \in \bF(t)[x]$ and~$p_d \neq 0$. Then
\[  p_d(x) = c p_d(x+m) q^{- dn} \quad {\rm and} \quad
    p_0(x) = c p_0(x+m). \]
Since~$P_0$ is irreducible and of positive degree in~$x$, $p_0$ is also nonzero.
We see that~$1 = c q^{-dn}$ and~$1 = c$
when comparing the leading coefficients in the above two equalities.
Consequently,~$q$ is a root of unity, a contradiction.
This proves that all irreducible factors of~$\den(a ) \num(a )$ are either in~$\bF(t)[x]$ or~$\bF(t)[y]$.
Therefore,~$a $ is a product of an element in~$\bF(t, x)$ and an element in~$\bF(t, y)$.

So we can write~$a  = a_1 \, a_2$ for some~$a_1$ in~$\bF(t, x)$
and~$a_2$ in~$\bF(t, y)$.
By~$\lsi(b) = \ltau(a)$, the
equation~$\si(z) = \ltau(a_2) z$ has a rational solution~$b$. Since~$\ltau(a_2)$
is a constant w.r.t.~$\si$, we conclude~$\ltau(a_2)=1$, for otherwise, $ \si(z) = \ltau(a_2) z$~would have no
rational solution. So~$b \in \bF(t, y)$ and~$a \in \bF(t, x)$.
\end{proof}
Similar to Lemmas~\ref{LM:fsw} and~\ref{LM:dq11}, the above lemma implies that
an $H$-solution~$h$ can be written
as a product of
a constant w.r.t.~$\si$ and~$\tau$,
a rational function, a hypergeometric term,
and a $q$-hypergeometric term.

We shall extend these lemmas to multivariate cases in Section~\ref{SECT:st}.
Before closing this section, we present three examples to illustrate calculations
involving compatibility conditions.
These calculations are useful in Section~\ref{SECT:st}.

\begin{example} \label{EX:intds}
Assume
$$u = \lde(f) + \lde(a) \, x  + b \quad {\rm and} \quad
v =  \lsi(f) \, a\, c,$$
where~$f, c \in \bF(t, x, y)^\times$, $a \in \bF(t,y)^\times$, and~$b \in \bF(t,x, y)$.
By the logarithmic derivative identity: for all~$r, s$ in~$\bF(t, x, y)^\times$, $\lde(r\, s)=\lde(r) + \lde(s)$,
we get
\[\lde(v) = \lde\circ \lsi(f) + \lde(a) + \lde(c).\]
Since~$\lde(a)$ is constant w.r.t.\ $\si$, and~$ \si \circ \lde = \lde \circ \si$, we have
\begin{align*}
\si(u)-u & = \si\circ \lde(f) - \lde(f) + \lde(a) + \si(b)-b\\
         & = \lde \circ \lsi(f) + \lde(a) + \si(b)-b.
\end{align*}
If~\eqref{EQ:ds1} holds, then~$\lde(c) {=} \si(b)-b$.
Hence, $\delta(c) {=} 0$ iff~$\si(b)=b$, i.e.,~$c \in \bF(x,y)$ iff~$b \in \bF(t,y)$.
\end{example}

\begin{example} \label{EX:intdq}
Assume
$u = \lde(f) + a$ and $w =  \ltau(f) \, b,$
where~$a \in \bF(t,x, y)$ and~$f, b  \in \bF(t, x, y)^\times$. If~\eqref{EQ:dq1} holds, then
a similar calculation as above
yields~$\lde(b)=\tau(a)-a$.
Hence, $\delta(b)=0$ iff~$\tau(a)=a$, i.e.,~$b \in \bF(x, y)$ iff~$a \in \bF(t, x)$.
\end{example}

\begin{example} \label{EX:intsq}
Assume~$v =   \lsi(f) \, a$ and~$w =  \ltau(f) \, b,$
where~$f, a, b \in \bF(t, x,y)^\times$. Applying~$\lsi, \ltau$ to~$w, v$, respectively,
we see that
$$\lsi(w) = \lsi \circ \ltau(f) \, \lsi(b), \quad \ltau(v) = \ltau \circ \lsi(f) \, \ltau(a).$$
If~\eqref{EQ:sq1} holds,
then~$\lsi(b) = \ltau(a)$, because~$\lsi \circ \ltau = \ltau \circ \lsi$.
Hence,~$\si(b)=b$ iff~$\tau(a)=a$, i.e., $b \in \bF(t, y)$ iff~$a \in \bF(t, x)$.
\end{example}


\section{Preparation lemmas}\label{SECT:mdf}
To extend Lemmas~\ref{LM:fsw}, \ref{LM:dq11}, and~\ref{LM:sq11} to multivariate cases,
we will proceed by induction on the number of variables. There arise
different expressions for a rational function in our  induction. Lemmas given in this section
will be used to  eliminate redundant indeterminates in these expressions.

We define a few additive subgroups of~$\bF(\vt, \vx, \vy)$ to avoid complicated expressions.
\[ L_i = \left\{ \lde_i(f) \, |\, f \in \bF(\vt, \vx, \vy)^\times \right\}, \,\, i=1, \ldots, l, \]
\[ M_i = \left\{ \sum_{j=1}^m  \lde_i(g_j) \, x_j \,|\, g_j \in \bF(\vt, \vy)^\times \right\},
\,\, i=1, \ldots, l. \]
For $i=1, \ldots, l$ and~$j=1, \ldots, m$, $M_{i,j}$ denotes the group
\[ \left\{ \sum_{k=1}^{j-1}  \lde_i(g_k) \, x_k
+ \sum_{k=j+1}^m  \lde_i(g_k) \, x_k  \, | \, g_k \in \bF( \vt, x_j, \vy)^\times \right\}. \]
Moreover, we set
\[N_i = L_i + M_i + \bF(\vt, \vy) \quad {\rm and} \quad N_{i,j}= L_i + M_{i,j} + \bF(\vt, x_j, \vy).\]
Let~$Z = \{t_1, \ldots, t_l, x_1, \ldots, x_m, y_1, \ldots, y_n\}$.
We will use an evaluation trick in the sequel.
Let~$Z^\prime = \{z_1, \ldots, z_s\}$ be a subset of~$Z$.
For $f \in \bF(\vt, \vx, \vy)^\times$,
there exist~$\xi_1,$ \ldots, $\xi_s$ in~$\bF$ s.t.~$f$ evaluated
at~$z_1=\xi_1$, \ldots,~$z_s = \xi_s$ is a well-defined and nonzero rational function~$f^\prime$.
We say that~$f^\prime$ is a {\em proper evaluation} of~$f$ w.r.t.~$Z^\prime$.
A proper evaluation can be carried out for finitely many rational functions as well.
In addition, we say that a rational function~$f$ is free of~$Z^\prime$ if it is free of every indeterminate in~$Z^\prime$.
\begin{remark} \label{RE:eval1}
If~$Z^\prime \subset Z$, $f \in L_i$ and~$t_i \notin Z^\prime$, then all proper evaluations of~$f$ w.r.t.~$Z^\prime$ are also in~$L_i$.
\end{remark}
In the next example, we illustrate two typical proper evaluations to be used later.
\begin{example} \label{EX:eval}
Let~$f = \lde_i(r)$ for some~$f, r \in \bF(\vt, \vx, \vy)^\times$.
Assume that both~$f(\vt, \bxi, \vy)$ and~$r(\vt, \bxi, \vy)$ are well-defined
and nonzero, where~$\bxi \in \bF^m$. Then~$f(\vt, \bxi, \vy)$ is still in~$L_i$.

Let~$g \in \bF(\vt, \vy)^\times$. Then~$ \delta_i(z)=gz$ has a rational solution in~$\bF(\vt,  \vy)^\times$
if it has a rational solution in~$\bF(\vt, \vx, \vy)^\times$.
This can also be shown by a proper evaluation.
\end{example}

The following lemma helps us merge rational expressions involving logarithmic
derivatives.

\medskip

\begin{lemma} \label{LM:dcomb} Let~$i \in \{1, \ldots, l\}$.
\begin{enumerate}
\item[(i)] Let~$Z_1, Z_2 \subset Z$ with~$Z_1 \cap Z_2 = \emptyset$.
If~$\bA$ is any subfield of~$\bF(\vt, \vx, \vy)$ whose elements
are free of~$t_i$ and free of~$Z_1 \cup Z_2$,
then
$$L_i+\bA(t_i) = (L_i+\bA(t_i, Z_1)) \cap (L_i + \bA(t_i, Z_2)).$$
\item[(ii)]
If~$d, e \in \{1, \ldots, m\}$ with~$d \neq e$,
then~$N_i = N_{i, d} \cap N_{i, e}.$
\end{enumerate}
\end{lemma}
\begin{proof}
To prove the first assertion, note that~$L_i+\bA(t_i)$ is a subset of~$(L_i+\bA(t_i, Z_1))\cap (L_i+\bA(t_i, Z_2))$.
Assume that~$a$ is in ~$(L_i+\bA(t_i, Z_1))\cap (L_i+\bA(t_i, Z_2))$. Then
there exist~$a_1 \in \bA(t_i, Z_1)$ and~$a_2 \in \bA(t_i, Z_2)$ s.t.
\[  a \equiv a_1 \mod L_i \quad {\rm and} \quad a \equiv a_2 \mod L_i. \]
Hence, $a_1-a_2 \in L_i$. Let~$Z_2^\prime = Z_2 \setminus \{t_i\}$, and~$a^\prime_2$ be a proper evaluation
of~$a_2$ w.r.t.~$Z_2^\prime$.
Then~$a_1-a^\prime_2$ is a proper evaluation of~$a_1-a_2$ w.r.t.~$Z_2^\prime$,
because~$a_1$ is free of~$Z_2^\prime$.
Thus,~$a_1-a^\prime_2$ belongs to~$L_i$ by Remark~\ref{RE:eval1}.
Since~$a_2^\prime$ is in~$\bA(t_i)$, $a_1$ is in~$L_i + \bA(t_i)$,
and so is~$a$.

To prove the second assertion, assume~$i=1$, $d=1$ and~$e=m$. Note that~$N_1\subset N_{1, 1}\cap N_{1, m}$,
because~$M_1$ is contained in $\left(M_{1,1} +  \bF(\vt, x_1, \vy) \right)\cap \left(M_{1,m} + \bF(\vt, x_m, \vy)\right)$.
It remains to show~$ N_{1,1} {\cap} N_{1,m} \subset N_1$. Let~$a {\in} N_{1,1} \cap N_{1,m}$. Then
\begin{eqnarray}
a & = \lde_1(f) + \left(\sum_{j=2}^{m-1} \lde_1(g_j) \, x_j  \right) +  \lde_1(g_m) \, x_m  + r  \label{EQ:expr1}\\
  & =  \lde_1(\tilde{f}) + \lde_1(\tilde{g}_1) \, x_1  +
  \left( \sum_{j=2}^{m-1} \lde_1(\tilde{g}_j) \, x_j \right) + \tilde{r}, \label{EQ:expr2}
\end{eqnarray}
where~$f, \tilde{f}\in \bF(\vt, \vx, \vy)$, $g_j, r\in \bF(\vt, x_1, \vy)$, $\tilde{g}_j, \tilde{r}\in \bF(\vt, x_m, \vy)$
and $f \tilde f g_j \tilde{g}_j \neq 0$. For all~$j$ with~$1 \le j \le m$,
let~$P_j$ be the polynomial part of~$a$ w.r.t.~$x_j$.
Then~$\deg_{x_j}P_j {\le} 1$ for all~$j$ with~$1 {\le} j {\le} m-1$ by Remark~\ref{RE:log} and~\eqref{EQ:expr2},
and~$\deg_{x_m} P_m {\le} 1$ by the same Remark and~\eqref{EQ:expr1}.

\smallskip \noindent
{\em Claim}.  Let~$b_j$ denote the coefficient of~$x_j$ in~$P_j$.
Then there exists $s_j \in \bF(\vt, \vy)$ s.t.~$b_j=\lde_1(s_j)$ for all~$j$ with~$1 \le j \le m$.

\smallskip \noindent
{\em Proof of Claim.} By~\eqref{EQ:expr1} and Remark~\ref{RE:log},  $b_1$ is the coefficient of~$x_1$ in the polynomial part
of~$r$ w.r.t.~$x_1$. So~$b_1$ is in~$\bF(\vt, \vy)$.
By~\eqref{EQ:expr2} and the same remark, $b_1 = \lde_1(\tilde{g}_1)$.
Let~$s_1$ be a proper evaluation of~$\tilde{g}_1$ w.r.t.~$x_m$. Then~$b_1=\lde_1(s_1)$ as~$b_1$ is free of~$x_m$.
By the same argument, $b_m=\lde_1(s_m)$ for some~$s_m$ in~$\bF(\vt, \vx)$.
By~\eqref{EQ:expr1} and~\eqref{EQ:expr2},~$b_j= \lde_1(g_j)  =  \lde_1(\tilde{g}_j)$
for all~$j$ with~$2 \le j \le m-1$. Let~$s_j$ be a proper evaluation of~$\tilde{g}_j$ w.r.t.~$x_m$.
Then~$\lde_j(g_j)=\lde_j(s_j)$, because~$g_j$ is free of~$x_m$. Hence,~$b_j = \lde_1(s_j)$.
The claim holds.

Set~$b=\sum_{j=1}^m b_j x_j$. Then
$a - b$ is in~$L_1+\bF(\vt, x_1, \vy)$ and~$L_1+\bF(\vt, x_m, \vy)$ by~\eqref{EQ:expr1}, ~\eqref{EQ:expr2} and the claim.
Thus,~$a-b$ is in~$L_1+\bF(\vt, \vy)$ by  the first assertion (setting~$Z_1=\{x_1\}$, $Z_2=\{x_m\}$,
and~$\bA=\bF(t_2, \ldots, t_l, \vy)$). By the claim, $b$ is in~$M_1$.
Thus,~$a$ is in~$L_1+M_1+\bF(\vt, \vy)$.
\end{proof}

We define a few multiplicative subgroups in~$\bF(\vt, \vx, \vy)^\times$.
Let~$G_j=\{\lsi_j(f) \, | \, f \in \bF(\vt, \vx, \vy)^\times \}$
for~$j=1, \ldots, m$. Similarly,
let~$H_k = \{ \ltau_k(f) \, | \, f \in \bF(\vt, \vx, \vy)^\times \}$
for~$k=1, \ldots, n$.
\begin{remark} \label{RE:eval2}
If~$Z^\prime {\subset} Z$,~$f {\in} H_k$ and~$y_k {\notin} Z^\prime$, then all proper evaluations of~$f$ w.r.t.~$Z^\prime$ are in~$H_k$.
The same holds for~$G_j$.
\end{remark}

The next lemma helps us merge rational expressions involving shift or
$q$-shift quotients.
\begin{lemma} \label{LM:sqcomb}
Let~$j \in \{1, \ldots, m\}$, $k \in \{1, \ldots, n\}$. Assume that~$Z_1$ and~$Z_2$
are disjoint subsets of~$Z$.
\begin{itemize}
\item[(i)] If~$\bA$ is any subfield of~$\bF(\vt, \vx, \vy)$~whose elements
are free of~$x_j$ and free of~$Z_1 \cup Z_2$, 
then
$$G_j \bA(x_j)^\times = \left(G_j \bA(x_j, Z_1)^\times \right)
\cap \left(G_j \bA(x_j, Z_2)^\times \right).$$
\item[(ii)] If~$\bA$ is any subfield of~$\bF(\vt, \vx, \vy)$~whose elements
are free of~$y_k$ and free of~$Z_1 \cup Z_2$, 
then
$$H_k \bA(y_k)^\times = \left(H_k \bA(y_k, Z_1)^\times \right)
\cap \left(H_k \bA(y_k, Z_2)^\times \right).$$
\item[(iii)] If~$\bA=\bF(\vt, \vy)$ and $\bB=\bF(\vx, \vy)$, then
$$G_j \bA^\times \bB^\times =  \left(G_j \bA^\times \bB(Z_1)^\times\right) \cap
\left(G_j \bA^\times \bB(Z_2)^\times\right).$$
\end{itemize}
\end{lemma}
\begin{proof}
The proofs of the first two assertions are similar to that of Lemma~\ref{LM:dcomb}~(i). So we only outline
the proof of the second assertion.
Clearly,
$$H_k \bA(y_k)^\times \subset \left(H_k \bA(y_k, Z_1)^\times \right)
\cap \left(H_k \bA(y_k, Z_2)^\times \right).$$
For an element~$a \in \left(H_k \bA(y_k, Z_1)^\times \right)
\cap \left(H_k \bA(y_k, Z_2)^\times \right)$, there exist~$a_1 \in \bA(y_k, Z_1)^\times$
and~$a_2 \in \bA(y_k, Z_2)^\times$ s.t.
\[ a \equiv a_1 \mod H_k \quad {\rm and} \quad a \equiv a_2 \mod H_k. \]
Using a proper evaluation, one sees that $a$ is in $H_k\bA(y_k)^\times$.

We present a detailed proof of the third assertion due to the presence
of both~$\bA$ and~$\bB$, though the idea goes along the same line as before.
It suffices to show that the intersection of
$G_j \bA^\times \bB(Z_1)^\times$ and~$G_j \bA^\times \bB(Z_2)^\times$
is a subset of~$G_j \bA^\times \bB^\times$.
Assume that~$a$ is in the intersection. Then
\begin{equation} \label{EQ:sqcomb}
a \equiv  a_1 b_1  \mod G_j \quad {\rm and} \quad a \equiv a_2 b_2 \mod G_j
\end{equation}
for some~$a_1, a_2$ in~$\bA^\times$, $b_1$ in~$\bB(Z_1)^\times$, and~$b_2$ in~$\bB(Z_2)^\times$.
Let~$Z_2^\prime = Z_2 \setminus \bB$, and~$c$ be a proper evaluation of~$a_1/(a_2b_2)$ w.r.t.~$Z_2^\prime$.
Then~$c \,b_1$ is a proper evaluation of~$a_1 b_1/(a_2b_2)$ w.r.t.~$Z_2^\prime$,
as~$b_1$ is free of~$Z_2^\prime$.
So~$c\, b_1$ is in~$G_j$ by Remark~\ref{RE:eval2}.
Since~$c$ is  in~$\bA^\times \bB^\times$, $b_1$ is in~$G_j\bA^\times \bB^\times$,
and so is~$a$.
\end{proof}

The next lemma says that some compatible rational functions belong to a common coset.
\begin{lemma} \label{LM:qcomb}
Let~$v_1, \ldots, v_m, w_1, \ldots, w_n {\in} \bF(\vt, \vx, \vy)^\times$.
Assume that the compatibility conditions in~\eqref{EQ:ss} and~\eqref{EQ:qq} hold.
\begin{enumerate}
\item[(i)] If~$v_j$ is in~$G_j \bF(\vt, \vy)^\times \bF(\vx, \vy)^\times$ for all~$j$ with~$1 \le j \le m$,
then there exists~$f \in \bF(\vt, \vx, \vy)$ s.t.\ each~$v_j$ is in  the coset~$\lsi_j(f) \bF(\vt, \vy)^\times \bF(\vx, \vy)^\times$.
\item[(ii)]
Let~$\bE$ be a subfield of~$\bF(\vt, \vx)$.
If~$w_k \in H_k \bE(\vy)^\times$ for all~$k$ with~$1 \le k \le n$,
then there exists~$f \in \bF(\vt, \vx, \vy)$ s.t.\ each~$w_k$ is in
the coset~$\ltau_k(f) \bE(\vy)^\times$.
\end{enumerate}
\end{lemma}
\begin{proof} We are going to show the second assertion. The first one can be proved
in the same fashion.

The second assertion clearly holds when~$n=1$.
Assume that~$n>1$ and the lemma holds for~$n-1$. Then there exist~$g  \in \bF(\vt, \vx, \vy)$
and~$b_1, \ldots, b_{n-1} \in \bE(\vy)$ s.t.~$w_k=\ltau_k(g) b_k$
for all~$k$ with~$1 \le k \le n-1$. Assume
\begin{equation} \label{EQ:gcomb1}
w_n = \ltau_n(g) \, a \quad \mbox{for some~$a \in \bF(\vt, \vx, \vy)$}.
\end{equation}
Then the compatibility conditions in~\eqref{EQ:qq} imply that the first-order
system~$\left\{\tau_k(z) =  \ltau_n(b_k) \, z \, | \,  k=1, \ldots, n-1 \right\}$
has a solution~$a$ in~$\bF(\vt, \vx, \vy)^\times$. It follows from
the hypothesis~$b_k \in \bE(\vy)$ for all~$k$ with~$1 \le k \le n-1$
that the above system has a solution~$a^\prime$ in~$\bE(\vy)^\times$.
Thus,~$a = c \, a^\prime$ for some constant $c$ w.r.t.~$\tau_1, \ldots, \tau_{n-1}$.
Consequently,~$c$ belongs to~$\bF(\vt, \vx, y_n)$.
On one hand,~\eqref{EQ:gcomb1} leads to
\begin{equation} \label{EQ:gcomb2}
w_n =  \ltau_n(g) \, c a^\prime.
\end{equation}
On the other hand, $w_n \in H_n \bE(\vy)^\times$ implies~$c= \ltau_n(s) \, r $ for some~$r$ in~$\bE(\vy)$
and~$s$ in~$\bF(\vt, \vx, \vy)$. Let~$Z^\prime = \{y_1, \ldots, y_{n-1}\}$, and let~$s^\prime$ and~$r^\prime$
be two proper evaluations of~$s$ and~$r$ w.r.t.~$Z^\prime$ at a point in~$\bF^{n-1}$, respectively.
Then~$c {=} \ltau_n(s^\prime) \, r^\prime$ since~$c$ is free of~$Z^\prime$. By~\eqref{EQ:gcomb2}, ~$w_n =  \ltau_n(s^\prime g) \, r^\prime a^\prime$.
Set~$f = s^\prime g$ and~$b_n {=} r^\prime a^\prime$. Then~$w_k = \ltau_k(f)\, b_k$  for all~$k$ with~$1 \le k \le n$,
as~$s^\prime$ is a constant w.r.t.~$\tau_1, \ldots, \tau_{n-1}$.
\end{proof}
\section{A structure theorem}  \label{SECT:st}
In this section, we extend Lemmas~\ref{LM:fsw}, \ref{LM:dq11} and~\ref{LM:sq11},
and then combine these results to a structure theorem on $\Delta$-compatible rational functions.

The first proposition extends Lemma~\ref{LM:fsw}.
\begin{prop} \label{PROP:ds}
Let~$u_1,$ \ldots, $u_l,$  $v_1,$ \ldots, $v_m$ be rational functions
in~$\bF(\vt, \vx, \vy)$ with~$v_1 \cdots v_m \neq 0$.
If the compatibility conditions in~\eqref{EQ:dd}, \eqref{EQ:ss} and~\eqref{EQ:ds} hold,
then there exist~$f$ in~$\bF(\vt, \vx, \vy)$, $a_1,$ \ldots, $a_m,$ $b_1,$ \ldots, $b_l$ in~$\bF(\vt, \vy)$,
and~$c_1,\ldots, c_m$ in~$\bF(\vx, \vy)$ s.t., for all~$i$ with~$1 \le i \le l$,
\[ u_i = \lde_i(f)  +  \lde_i(a_1) \, x_1
+ \cdots +   \lde_i(a_m) \, x_m  + b_i ,  \]
and, for all~$j$ with~$1 \le j \le m$,
$$v_j = \lsi_j(f) \, a_j\, c_j .$$
Moreover, the sequence $b_1, \ldots, b_l,$ $c_1,$ \ldots, $c_m$ is compatible w.r.t.~$\{\delta_1, \ldots, \delta_l,
\si_1, \ldots, \si_m\}$.
\end{prop}
\begin{proof}
First, we consider the case in which~$l =1$ and~$m$ arbitrary.
The proposition holds when~$m=1$ by Lemma~\ref{LM:fsw}. Assume that~$m>1$
and the proposition holds for the values lower than~$m$.
Applying the induction hypothesis to~$t_1,$ $x_1,$ \ldots, $x_{m-1}$
and to~$t_1, x_2, \ldots, x_m$, respectively, we see that  both~$u_1 \in N_{1,m}$
and~$u_1 \in N_{1,1}$. Since~$m>1$, $u_1 \in N_1$ by Lemma~\ref{LM:dcomb}~(ii). Hence,
\[ u_1 = \lde_1(f)  +  \lde_1(a_1) \, x_1
+ \cdots + \lde_1(a_m) x_m  + b_1 \]
for some~$f \in \bF(t_1, \vx, \vy)$ and~$a_1, \ldots, a_m, b_1 \in \bF(t_1, \vy)$.
Assume that~$v_j = \lsi_j(f) \,  a_j c_j $.
Then~$c_1$, \ldots,~$c_m$ are in~$\bF(\vx, \vy)$ by the compatibility conditions
in~\eqref{EQ:ds} (see Example~\ref{EX:intds}). The proposition
holds for~$l=1$ and~$m$ arbitrary.

Second, we show that the proposition holds for all~$l$ and~$m$ by induction on~$l$.
It holds if~$l=1$ by the preceding paragraph. Assume that~$l>1$ and that
the proposition holds for the values lower than~$l$.
Applying the induction hypothesis to~$t_1, \ldots, t_{l-1}, \vx$ and
to~$t_2, \ldots, t_l, \vx$, respectively, we have
$$v_j \in \left(G_j \bA^\times \bB(Z_1)^\times\right) \cap
\left(G_j \bA^\times \bB(Z_2)^\times\right),$$
where~$\bA{=}\bF(\vt, \vy)$, $\bB{=}\bF(\vx, \vy)$, $Z_1{=}\{t_l\}$, and~$Z_2{=}\{t_1\}$.
We see that~$v_j \in G_j \bA^\times \bB^\times$ by Lemma~\ref{LM:sqcomb}~(iii).
So~$v_j \in \lsi_j(f) \bA^\times \bB^\times$ for some~$f$ in~$\bF(\vt, \vx, \vy)$ by Lemma~\ref{LM:qcomb}~(i).
Thus,
$$v_j = \lsi_j(f) \, a_j \, c_j,$$
 where~$a_j \in \bA$, $c_j \in \bB$ and~$j=1, \ldots, m$.
Assume that, for all~$i$ with~$1 \le i \le l$,
$u_i =  \lde_i(f) + \sum_{j=1}^m \lde_i(a_j) \,  x_j + b_i.$
All the~$b_i$'s belong to~$\bF(\vt, \vy)$ by the compatibility conditions
in~\eqref{EQ:ds} (see Example~\ref{EX:intds}).
The sequence $b_1, \ldots, b_l,$ $c_1,$ \ldots,~$c_m$ is compatible because of~\eqref{EQ:dd}, \eqref{EQ:ss} and~\eqref{EQ:ds}.
\end{proof}
The second proposition extends Lemma~\ref{LM:dq11}.
\begin{prop} \label{PROP:dq}
Let~$u_1,$ \ldots, $u_l,$ $w_1,$ \ldots, $w_n$ be
rational functions in~$\bF(\vt, \vx, \vy)$ with~$w_1 \cdots w_n \neq 0$.
Assume that the compatibility conditions~\eqref{EQ:dd}, \eqref{EQ:qq} and~\eqref{EQ:dq} hold.
Then there exist~$f$ in~$\bF(\vt, \vx, \vy)$, $a_1, \ldots, a_l$ in~$\bF(\vt, \vx)$,
and~$b_1, \ldots, b_n$ in~$\bF(\vx, \vy)$ s.t.
$$u_i = \lde_i(f) + a_i \quad \mbox{ and } \quad w_k = \ltau_k(f) \, \, b_k$$
for all~$i$ with~$1 \le i \le l$ and~$k$ with~$1 \le k \le n$.
Moreover, the sequence~$a_1, \ldots, a_l$, $b_1,$ \ldots, $b_n$
is compatible w.r.t. the set~$\{\delta_1, \ldots, \delta_l, \tau_1, \ldots, \tau_n\}$.
\end{prop}
The proof of this proposition goes along the same line as in that of Proposition~\ref{PROP:ds}.

The last proposition extends Lemma~\ref{LM:sq11}.
\begin{prop} \label{PROP:sq}
Let~$v_1, \ldots, v_m, w_1, \ldots, w_n$ be rational functions in~$\bF(\vt, \vx, \vy)^\times$.
Assume that the compatibility
conditions in~\eqref{EQ:ss}, \eqref{EQ:qq} and~\eqref{EQ:sq} hold. Then
there exist a rational function~$f$ in~$\bF(\vt, \vx, \vy)$, $a_1, \ldots, a_m$ in~$\bF(\vt, \vx)$, and~$b_1, \ldots, b_n$ in~$\bF(\vt, \vy)$
s.t., for all~$j$ with~$1 \le j \le m$ and~$k$ with~$1 \le k \le n$,
\[ v_j = \lsi_j(f) \, a_j  \quad {\rm and} \quad
   w_k = \ltau_k(f) \, b_k.  \]
Furthermore, the sequence~$a_1, \ldots, a_m$, $b_1,$ \ldots, $b_n$  is compatible w.r.t.~$\{\si_1, \ldots, \si_m, \tau_1, \ldots, \tau_n\}$.
\end{prop}
\begin{proof}
First, we consider the case, in which~$m =1$ and~$n$ arbitrary.
We proceed by induction on~$n$. The proposition holds when~$n=1$ by Lemma~\ref{LM:sq11}.
Assume that~$n{>}1$, and the proposition holds for the values lower than~$n$.
Applying the induction hypothesis to~$x_1, y_1, \ldots, y_{n-1}$ and to~$x_1, y_2, \ldots, y_n$,
respectively, we get~$v_1 \in G_1 \bF(\vt, x_1,  y_n)^\times \cap G_1 \bF(\vt, x_1,  y_1)^\times$.
Setting~$\bA{=}\bF(\vt)$, $Z_1{=}\{y_n\}$ and~$Z_2{=}\{y_1\}$ in Lemma~\ref{LM:sqcomb}~(i),
we see that~$v_1 \in G_1\bF(\vt, x_1)^\times$,
which, together with the definition of~$G_1\bF(\vt, x_1)^\times$,
there exist~$f$ in~$\bF(\vt, x_1, \vy)$ and~$a$ in~$\bF(\vt, x_1)$ s.t.~$v_1=\lsi_1(f)\, a$.
Assume that~$w_k =\ltau_k(f) \, b_k$
for some~$b_k$ in~$\bF(\vt, x_1, \vy)$ and for all~$k$ with~$1 \le k \le n$.
By~\eqref{EQ:sq},~$\si_1(b_k){=}b_k$, i.e., $b_k \in \bF(\vt, \vy)$ (see Example~\ref{EX:intsq}). The proposition holds
for~$m=1$ and~$n$ arbitrary.

Second, assume that~$m>1$ and the proposition holds for values
lower than~$m$ and arbitrary~$n$.
Applying this induction hypothesis to~$x_1, \ldots, x_{m-1}, \vy$
and to~$x_2, \ldots, x_m, \vy$, respectively, we have
$$ w_k \in \left(H_k \bA(y_k, Z_1)^\times \right)
\cap \left(H_k \bA(y_k, Z_2)^\times \right),$$
where~$\bA {=}\bF(\vt, y_1,{\ldots}, y_{k-1}, y_{k+1},{\ldots}, y_n)$, $Z_1$ and~$Z_2$ are equal to~$\{x_m\}$
and~$\{x_1\}$, respectively.
Thus,~$w_k \in H_k \bA(y_k)^\times $ by Lemma~\ref{LM:sqcomb}~(ii), and~$w_k \in \ltau_k(f) \bA(y_k)^\times$
for some~$f$ in~$\bF(\vt, \vx, \vy)$ by Lemma~\ref{LM:qcomb}~(ii).
Let~$w_k = \ltau_k(f)\, b_k,$
where~$b_k$ is in~$\bA(y_k)^\times$, and~$k=1,$ \ldots, $n$.
Let~$a_j = v_j/\lsi_j(f)$ for all~$j$ with~$1 \le j \le m$.
Then~$\tau_k(a_j)=a_j$ for all~$k$ with~$1 \le k \le n$ and~$j$ with~$1 \le j \le m$
by the compatibility conditions in~\eqref{EQ:sq} (see Example~\ref{EX:intsq}).
Hence, all the~$a_j$'s
are in~$\bF(\vt, \vx)$.
The sequence~$a_1, \ldots, a_m$, $b_1,$ \ldots, $b_n$  is compatible because of~\eqref{EQ:ss}, \eqref{EQ:qq} and~\eqref{EQ:sq}.
\end{proof}

Now, we present a theorem describing the structure of
compatible rational functions.
\begin{theorem} \label{TH:dsq}
Let
\begin{equation} \label{EQ:comp}
u_1, \ldots, u_l, \, \, v_1, \ldots, v_m, \, \, w_1, \ldots, w_n
\end{equation}
be a sequence of rational functions in~$\bF(\vt, \vx, \vy)$. If the sequence is $\Delta$-compatible,
then there exist~$f$ in~$\bF(\vt, \vx, \vy)$, $\alpha_1,$ \ldots, $\alpha_m,$ $\beta_1,$ \ldots, $\beta_l$
in~$\bF(\vt)$, $\lambda_1,$ \ldots, $\lambda_m$ in~$\bF(\vx)$, and~$\mu_1,$ \ldots, $\mu_n$ in~$\bF(\vy)$
s.t., for all~$i$ with~$1 \le i \le l$,
\begin{equation} \label{EQ:u}
u_i = \lde_i(f) + \lde_i(\alpha_1) \,  x_1
+ \cdots + \lde_i(\alpha_m) \, x_m + \beta_i,
\end{equation}
for all~$j$ with~$1 \le j \le m$, and, for all~$k$ with~$1 \le k \le n$,
\begin{equation} \label{EQ:vw}
 v_j = \lsi_j(f) \, \alpha_j \lambda_j \quad
{\rm and} \quad
w_k = \ltau_k(f) \, \mu_k.
\end{equation}
Moreover, the sequence $\beta_1,$ \ldots, $\beta_l$, $\lambda_1$, \ldots, $\lambda_m$, $\mu_1$, \ldots, $\mu_n$
is $\Delta$-compatible.
\end{theorem}
\begin{proof}
By Propositions~\ref{PROP:dq} and~\ref{PROP:sq},
\[  w_k  = \ltau_k(g^\prime) \, a^\prime_k  =  \ltau_k(\tilde g) \, \tilde{a}_k \]
for some $g^\prime, \tilde g \in \bF(\vt, \vx, \vy)$, $a^\prime_k \in \bF(\vx, \vy)$, and $\tilde{a}_k \in \bF(\vt, \vy)$ with~$1 \le k \le n$.
Set~$Z_1=\{t_1, \ldots, t_l\}$,~$Z_2=\{x_1, \ldots, x_m\}$, and $\bA=\bF(y_1, \ldots, y_{k-1}, y_{k+1}, \ldots, y_n)$
in Lemma~\ref{LM:sqcomb}~(ii).
Then the lemma implies that there exist~$\mu_k$ in~$\bF(\vy)$ and~$g_k$ in~$\bF(\vt, \vx, \vy)$
s.t.~$w_k = \ltau_k(g_k) \, \mu_k.$
Setting~$\bE=\bF$ in the second assertion of Lemma~\ref{LM:qcomb}, we may further assume that
all the~$g_k$'s are equal to a rational function, say~$g$.
Let
$$u_i = \lde_i(g) + r_i \,\, (1 \le i \le l) \,\, {\rm and} \,\,
v_j = \lsi_j(g) \, s_j \,\, (1 \le j \le m).$$
Then the compatibility conditions in~\eqref{EQ:dq} imply that the~$r_i$'s
are in~$\bF(\vt, \vx)$ (see Example~\ref{EX:intdq}). Similarly, those conditions in~\eqref{EQ:sq} imply
that the~$s_j$'s are in~$\bF(\vt, \vx)$ (see Example~\ref{EX:intsq}). Furthermore, $r_1, \ldots, r_l, s_1, \ldots, s_m$
are compatible w.r.t.\ the set~$\{\delta_1, \ldots, \delta_l,  \si_1, \ldots, \si_m\}$.
By Proposition~\ref{PROP:ds}, we get
\[ r_i = \lde_i(b) + \lde_i(\alpha_1) \, x_1
+ \cdots + \lde_i(\alpha_m) \, x_m + \beta_i, \]
and~$s_j =  \lsi_j(b) \, \alpha_j \lambda_j$
for some~$b$ in~$\bF(\vt, \vx)$,
$\alpha_j, \beta_i$ in~$\bF(\vt)$, $\lambda_j$ in~$\bF(\vx)$,
$1 \le i \le l$, and~$1 \le j \le m$.
Note that~$b$ belongs to~$\bF(\vt, \vx)$. Setting~$f = g b$, we get the desired form
for~$u_i$'s, $v_j$'s and~$w_k$'s.
The compatibility of the sequence~$\beta_1,$ \ldots, $\beta_l$, $\lambda_1$, \ldots, $\lambda_m$, $\mu_1$, \ldots, $\mu_n$
follows from that of~$u_1,$ \ldots, $u_l$, $v_1$, \ldots, $v_m$, $w_1$, \ldots, $w_n$.
\end{proof}


With the notation introduced in Theorem~\ref{TH:dsq}, we say that the sequence:
\begin{equation} \label{EQ:repr}
f, \, \alpha_1, \, \ldots, \, \alpha_m, \, \beta_1, \, \ldots, \, \beta_l, \lambda_1, \,
\ldots, \, \lambda_m, \mu_1, \, \ldots, \, \mu_n
\end{equation}
is a {\em representation} of  $\Delta$-compatible rational
functions given in~\eqref{EQ:comp} if the equalities
in~\eqref{EQ:u} and~\eqref{EQ:vw} hold.

A rational function~$\bF(\vt, \vx, \vy)$ is said to be {\em nonsplit} w.r.t.~$\vt$
if its denominator and numerator have no irreducible factors
in~$\bF[\vt]$. Similarly, we define the notion of nonsplitness w.r.t.~$\vx$ or~$\vy$.
Let~$\prec$ be a fixed monomial ordering on~$\bF[\vt, \vx, \vy]$.
A nonzero rational function in~$\bF(\vt, \vx, \vy)$ is said to be {\em monic}
w.r.t.~$\prec$ if its denominator and numerator are both monic w.r.t.~$\prec$.
A representation~\eqref{EQ:repr} of $\Delta$-compatible rational
functions in~\eqref{EQ:comp} is said to be {\em standard} w.r.t.~$\prec$
if
\begin{enumerate}
\item[(i)] $f$ is nonsplit w.r.t.~$\vt$, $\vx$, and~$\vy$, that is, the nontrivial irreducible
factors of~$\den(f)\num(f)$ are neither in~$\bF[\vt]$, nor in~$\bF[\vx]$, nor in~$\bF[\vy]$;
\item[(ii)] both~$f$ and~$\alpha_j$ are monic w.r.t.~$\prec$, $j=1,2, \ldots, m$.
\end{enumerate}
Assume that the sequence~\eqref{EQ:repr} is a representation of~\eqref{EQ:comp}.
Factor~$f {=} f_1 f_2 f_3 f_4$, where~$f_1$ is monic and nonsplit w.r.t.~$\vt$,~$\vx$ and~$\vy$,
$f_2$ is in~$\bF(\vt),$ $f_3$ in~$\bF(\vx),$  and~$f_4$ in~$\bF(\vy)$.
Set~$\alpha_j {=} c_j \alpha_j^\prime$, where~$c_j \in \bF$, and~$\alpha_j^\prime$ is monic. Then
\[ f_1, \, \alpha_1^\prime, \, \ldots,  \, \alpha_m^\prime,  \, \beta_1 + \lde_1(f_2),
\, \ldots, \, \beta_l + \lde_l(f_2), \]
\[ \lsi_1(f_3) c_1\lambda_1, \, \ldots, \,  \lsi_m(f_3) c_m\lambda_m, \,
\ltau_1(f_4)  \mu_1, \,  \ldots,  \,  \ltau_n(f_4) \mu_n \]
is also a representation of~\eqref{EQ:comp}. This proves the existence
of standard representations. Its uniqueness
follows from the uniqueness of factorization of rational functions.
\begin{cor} \label{COR:nr}
A $\Delta$-compatible sequence has a unique standard representation w.r.t.\ a given
monomial ordering.
\end{cor}
\section{Algorithms and applications} \label{SECT:ap}
In this section, we discuss how to compute a representation of compatible rational functions,
and present two applications in analyzing $H$-solutions.
Let us fix a monomial ordering on~$\bF[\vt, \vx, \vy]$ for standard representations.

Let the sequence given in~\eqref{EQ:comp} be $\Delta$-compatible.
We compute a representation of the sequence in the form of~\eqref{EQ:repr}.

First, we compute~$\mu_1(\vy), \ldots, \mu_n(\vy)$ in the sequence~\eqref{EQ:repr}.
By gcd-computation, we write~$w_k = a_k b_k$, where~$a_k$ is nonsplit w.r.t.~$\vy$, $b_k$ is in~$\bF(\vy)$, and~$k=1$, \ldots, $n$.
By Theorem~\ref{TH:dsq},~$w_k = \ltau_k(f) \, \mu_k$, where~$f$ is nonsplit w.r.t.~$\vy$ and~$\mu_k$ is in~$\bF(\vy)$.
Thus,~$b_k = c_k \mu_k$ for some~$c_k \in \bF^\times$.

To determine~$c_k$, write~$a_k = \ltau_k(g_k) \, r_k$, where~$g_k$ and~$r_k$ are in~$\bF(\vt, \vx, \vy)$
with~$r_k$ being $\tau_k$-reduced.
By the two expressions of~$w_k$,~$c_k r_k = \ltau_k(f/g_k)$.
Since~$a_k$ is nonsplit w.r.t.~$\vy$ and~$r_k$ is $\tau_k$-reduced, $g_k$ can be chosen to be nonsplit w.r.t.~$\vy$, and so is~$f /g_k$.
Thus,~$f/g_k$ is free of~$y_k$, because~$c_k r_k$ is $\tau_k$-reduced.
Accordingly,~$c_kr_k{=}1$ and~$\mu_k {=} r_k b_k$. As a byproduct, we obtain~$g_k$ with~$\ltau_k(f) = \ltau_k(g_k)$.

Second, we compute~$\alpha_1$, $\ldots$, $\alpha_m$ and~$\lambda_1, \ldots, \lambda_m$.
Assume that~$j$ is an integer with~$1 \le j \le m$.
By gcd-computation, we write~$v_j = s_j a_j b_j $, where~$s_j$ is nonsplit w.r.t.~$\vt$ and~$\vx$,
$a_j$ is in~$\bF(\vt)$, and~$b_j$ in~$\bF(\vx)$. Moreover, set~$a_j$ to be monic.
By Theorem~\ref{TH:dsq},~$v_j = \lsi_j(f) \, \alpha_j \lambda_j$, where~$f$ is nonsplit w.r.t.~$\vt$ and~$\vx$,
$\alpha_j$ is a monic element in~$\bF(\vt)$, and~$\lambda_j$ is in~$\bF(\vx)$.
Hence,~$a_j=\alpha_j$ and~$b_j = c_j \lambda_j$ for some~$c_j \in \bF^\times$.
As in the preceding paragraph, we write~$s_j = \lsi_j( g^\prime_j) \, r_j$ with~$r_j$ being $\si_j$-reduced.
Then~$c_j r_j = \lsi_j(f/g^\prime_j)$. Since~$c_jr_j$ is $\si_j$-reduced, $c_jr_j=1$.
Hence,~$\lambda_j = r_j b_j$. As a byproduct, we find~$g_j^\prime$ with~$\lsi_j(f) = \lsi_j(g^\prime_j)$.

Third, we compute~$f$.
Note that~$f$ is a nonzero rational solution of the system~$\{\si_j(z)= \lsi_j(g_j^\prime) \, z , \tau_k(z)=\ltau_k(g_k) \,z\},$
where~$1 \le j \le m$, $1 \le k \le n$, and~$g_j^\prime, g_k$ are obtained in the first two steps.
So~$f$ can be computed by several methods, e.g., the method in the proof
of~\cite[Proposition 3]{LabahnLi2004}.

At last, we set~$\beta_i = u_i - \lde_i(f) - \sum_{j=1}^m  \lde_i(\alpha_j) \, x_j$,
for all~$i$ with $1\leq i \leq l$.
Using~$v_j = \lsi_j(f)\, \alpha_j \lambda_j$ and~$w_k=\ltau_k(f) \, \mu_k$
and the compatibility conditions in~\eqref{EQ:ds} and~\eqref{EQ:dq}, we see
that all the~$\beta_i$'s are in~$\bF(\vt)$, as required.
\begin{example} \label{EX:comp}
Consider the case~$l=m=n=1$.
Let~$u$, $v$ and~$w$ be compatible rational functions, where
\begin{align*}
    u & = \frac{(4t+2x+y^2)(t+1) + (t+x+1)(t+x)(2t+y^2)}{(t+1)(t+x)(2t+y^2)},\\
    v & = \frac{2(2x+3)( x+1)(t+1)(t+x+1)(5x+y)}{(5x+y+5)(t+x)}, \\
    w & = \frac{(5x+y)(2t+q^2y^2)(1+qy)}{(5x+qy)(2t+y^2)}.
\end{align*}
A representation of~$u, v, w$ is of the form
$$ \left(\frac{(2t+y^2)(t+x)}{5x+y}, \,  t+1, \, 1, \, 2(2x+3)(x+1), \, q y + 1 \right).$$
\end{example}

From now on, we assume that our ground field~$\bF$ is algebraically closed.
In general, $\Delta$-extensions of~$\bF(\vt, \vx, \vy)$ are rings.
We recall that an $H$-solution over
$\bF(\vt,\vx,\vy)$ is a nonzero solution of system (\ref{EQ:sys}) and, given a finite number of $H$-solutions,
there is a $\Delta$-extension of~$\bF(\vt,\vx,\vy)$ containing these $H$-solutions and their inverses.
The ring of constants of this $\Delta$-extension is equal to~$\bF$ by Theorem~2 in~\cite{Bronstein2005}.
We will only encounter finitely many pairwise dissimilar $H$-solutions.
Hence, it makes sense to multiply and invert them in some $\Delta$-extension, which will not be specified
explicitly if no ambiguity arises.
All $H$-solutions we consider will be over~$\bF(\vt,\vx,\vy)$. Denote by~$\bzero_s$ and~$\bone_s$
the sequences consisting of~$s$ $0$'s and of~$s$ $1$'s, respectively.

An $H$-solution is said to be a {\em symbolic power} if its certificates are of the form
\begin{equation} \label{EQ:sp}
 \sum_{j=1}^m x_j \lde_1(\alpha_j), \, \ldots, \,
   \sum_{j=1}^m x_j \lde_l(\alpha_j), \, \alpha_1, \, \ldots,  \, \alpha_m, \, \bone_n,
\end{equation}
where~$\alpha_1, \ldots, \alpha_m$ are monic elements in~$\bF(\vt)^\times$.
It is easy to verify that such a sequence is $\Delta$-compatible.
Such a symbolic power is denoted~$\alpha_1^{x_1} \cdots \alpha_m^{x_m}$.
The monicity of the~$\alpha_i$'s excludes the case, in which some~$\alpha_i$ is a
constant different from one.
By an $E$-solution, we mean an $H$-solution whose certificates are of the form
$\beta_1, \ldots, \beta_l, \bone_{m+n}$,
where~$\beta_1, \ldots, \beta_l$ are in~$\bF(\vt)$.
An $E$-solution is a hyperexponential function w.r.t.~the derivations, and a constant w.r.t.\ other
operators.
By a $G$-solution,
we mean an $H$-solution whose certificates are of the form~$\bzero_l, \lambda_1,$ \ldots, $\lambda_m, \bone_n$,
where~$\lambda_1,$ \ldots, $\lambda_m$ are in~$\bF(\vx)^\times$.
A $G$-solution is a hypergeometric term w.r.t.\ the shift operators, and a constant w.r.t.\ other operators.
Similarly, by a $Q$-solution, we mean an $H$-solution whose certificates are of the
form~$\bzero_l,$ $\bone_m,$ $\mu_1,$ \ldots, $\mu_n$,
where~$\mu_1,$ \ldots, $\mu_n$ are in~$\bF(\vy)^\times$.
A $Q$-solution is a $q$-hypergeometric term w.r.t.\ the $q$-shift operators, and a constant w.r.t.\
other operators.

The next proposition describes a multiplicative decomposition of $H$-solutions.
\begin{prop} \label{PROP:hyper}
An $H$-solution is a product of an element in~$\bF^\times$, a rational function in~$\bF(\vt, \vx, \vy)$,
a symbolic power, an $E$-solution, a $G$-solution, and a $Q$-solution.
\end{prop}
\begin{proof}
Let~$h$ be an $H$-solution. Then its certificates are compatible.
By Theorem~\ref{TH:dsq}, the certificates have a standard
representation~$f,$ $\alpha_1,$ \ldots, $\alpha_m,$ $\beta_1,$ \ldots, $\beta_l,$ $\lambda_1,$ \ldots, $\lambda_m,$
$\mu_1,$ \ldots, $\mu_n.$
Moreover, the following three sequences:
$$\beta_1, \ldots, \beta_l, \bone_{m+n}; \quad \bzero_l, \lambda_1, \ldots, \lambda_m, \bone_n; \quad
\bzero_l, \bone_m, \mu_1, \ldots, \mu_n$$ are $\Delta$-compatible, respectively.
Hence, there exist an $E$-solution~$\cE$, a $G$-solution~$\cG$, and a $Q$-solution~$\cQ$
s.t.~their certificates are given in the above three sequences, respectively.
It follows from Theorem~\ref{TH:dsq} that~$h$ and the product~$f\alpha_1^{x_1} \cdots \alpha_m^{x_m} \cE \cG \cQ$
have the same certificates. So they differ by a multiplicative constant, which is in~$\bF$.
\end{proof}

The $H$-solution in Example~\ref{EX:comp} can be decomposed as
$$\frac{(2t+y^2)(t+x)}{5x+y}\, (t+1)^x \, \exp(t)\,  (2x+1)! \, \Gamma_q(1+qy),$$
where~$\Gamma_q(1+qy)$ is a $Q$-solution with certificates~$0, 1, 1+qy$.

The next proposition characterizes rational $H$-solutions via their standard representations.
\begin{prop}
\label{PROP:rat}
Let $\cP$ be a symbolic power, $\cE$ an $E$-solution, $\cG$ a $G$-solution and
$\cQ$ a $Q$-solution. Then $\cP\cE\cG\cQ$ is in $\bF(\vt,\vx,\vy)$ iff
$\cP\in \bF$, $\cE\in \bF(\vt)$, $\cG\in \bF(\vx)$ and $\cQ\in\bF(\vy)$.
\end{prop}
\begin{proof} $(\Leftarrow)$ Clear.

$(\Rightarrow)$
Assume that~$f$ is rational and equal to~$\cP\cE\cG\cQ$, where $\cP, \cE, \cG, \cQ$ are a symbolic power, an $E$-,  a $G$-, and a $Q$-solution, respectively.
Suppose that the certificates of~$\cP$ are given in~\eqref{EQ:sp}.
Applying~$\lde_i$ to~$f$, $i=1, \ldots, l$, we see that
$$\lde_i(f) = \sum_{j=1}^m \lde_i(\alpha_j) x_j + \lde_i(\cE).$$
Comparing the polynomial parts of the left and right hand-sides of the above equality w.r.t.~$x_j$,
we see that~$\lde_i(\alpha_j)=0$ by Remark~\ref{RE:log} and $\lde_i(\cE) \in \bF(\vt)$
for all~$i$ and~$j$. Hence, all the~$\alpha_j$'s are in~$\bF$, and, consequently, all the~$\alpha_j$'s
are equal to one as they are monic. Hence,~$\cP$ is in~$\bF$. Moreover,
\[\lde_i(f)=\lde_i(\cE) \quad \mbox{for all~$i$ with~$1 \le i \le l$}.\]
Let~$g$ be a proper evaluation of~$f$ w.r.t.~$\vx$ and~$\vy$.
Then
$$\lde_i(g) = \lde_i(\cE) \quad \mbox{for all~$i$ with~$1 \le i \le l$},$$
since~$\lde_i(\cE)$ is in~$\bF(\vt)$. Hence,
$\lde_i(\cE/g)=0,$
$\lsi_j(\cE/g)=1,$
and~$\ltau_k(\cE/g)=1,$
where~$1 {\le} i {\le} l$, $1 {\le} j {\le} m$, and~$1 {\le} k {\le} n$.
We conclude that~$\cE =c g$ for some~$c \in \bF$.
So~$\cE$ is in~$\bF(\vt)$.

Applying~$\lsi_j$ and~$\ltau_k$ to~$f$ leads to~$\lsi_j(f) = \lsi_j(\cG)$,
and~$\ltau_k(f) = \ltau_k(\cQ)$,   respectively.
One can show that~$\cG$ is in~$\bF(\vx)$ and~$\cQ$ is in~$\bF(\vy)$ by similar arguments.
\end{proof}
Now, we consider how to determine whether a finite number of $H$-solutions are
algebraically dependent over~$\bF(\vt, \vx, \vy)$.
Let $h_1,\cdots,h_s$ be $H$-solutions. By Proposition~\ref{PROP:hyper},
\begin{equation} \label{EQ:prod}
h_i \equiv  \cP_i \cE_i \cG_i\cQ_i \mod \bF(\vt, \vx, \vy)^\times, \quad i=1, \ldots, s,
\end{equation}
where
$\cP_i,\cE_i,\cG_i,\cQ_i$ are
a symbolic power, an $E$-solution, a $G$-solution, and a $Q$-solution, respectively.
\begin{cor} \label{COR:ad}
Let~$h_1, \ldots, h_s$ be $H$-solutions s.t.\ all the congruences in~\eqref{EQ:prod} hold.
Then they are algebraically dependent over~$\bF(\vt, \vx, \vy)$
iff there exist integers~$\omega_1,$ \ldots $\omega_s$, not all zero, s.t.\
$\cP_1^{\omega_1}\cdots\cP_s^{\omega_s}$ is in~$\bF$, $\cE_1^{\omega_1}\cdots\cE_s^{\omega_s}$ in~$\bF(\vt)$,
$\cG_1^{\omega_1}\cdots\cG_s^{\omega_s}$ in~$\bF(\vx)$ and $\cQ_1^{\omega_1}\cdots\cQ_s^{\omega_s}$ in~$\bF(\vy)$.
\end{cor}
\begin{proof}
It follows from~\cite[Corollary 4.2]{LiWuZheng2007} that $h_1,\cdots,h_s$ are algebraically
dependent over~$\bF(\vt, \vx, \vy)$ iff
there exist integers~$\omega_1,$ \ldots $\omega_s$, not all zero, s.t.~$h_1^{\omega_1} \cdots h_s^{\omega_s}$
is in~$\bF(\vt, \vx, \vy)$. The corollary follows from~\eqref{EQ:prod} and Proposition~\ref{PROP:rat}.
\end{proof}

By the above corollary, one may determine the algebraic dependence of~$h_1$, \ldots, $h_s$
using the decompositions in Proposition~\ref{PROP:hyper}.
By gcd-computation, one can find all nonzero integer vectors~$(\omega_1, \ldots \omega_s)$ s.t.
$\cP_1^{\omega_1}\cdots\cP_s^{\omega_s}$ is in~$\bF$.
According to~\cite{Sing2007}, one can find all nonzero integer vectors~$(\omega_1, \ldots \omega_s)$
s.t.~$\cE_1^{\omega_1}\cdots\cE_s^{\omega_s}\in \bF(\vt)$ by
seeking rational number solutions of a linear homogeneous system over~$\bF$.
Computing all nonzero integer vectors~$(\omega_1, \ldots \omega_s)$
s.t.~$\cG_1^{\omega_1}\cdots\cG_s^{\omega_s} \in \bF(\vx)$ reduces to the following subproblem:
given~$c_1, \ldots, c_s \in \bF^\times$,
compute integers~$\omega_1,$ \ldots $\omega_s$, not all zero, with
$c_1^{\omega_1}\cdots c_s^{\omega_s} = 1$ (see~\cite{Sing2007}).
Algorithms for tackling this subproblem and related discussions are
contained in~\cite[\S 7.3]{Kauers2005} and the references given there.
We are trying to develop an algorithm that finds integers~$\omega_1,$ \ldots $\omega_s$, not all zero,
s.t.~$\cQ_1^{\omega_1} \cdots \cQ_s^{\omega_s}$ belongs to~$\bF(\vy)$.

The reader is referred to~\cite{CFFLpre2011} for an extended version of this paper, which contains
a short proof of Fact~\ref{FA:dn} and a proof of Proposition~\ref{PROP:dq}.
A Maple implementation is being written for decomposing $H$-solutions.
We shall apply the structure theorem to study the existence of telescopers
in the mixed cases in which any two of differential, shift and~$q$-shift operators appear.

\medskip

\noindent {\bf Acknowledgments.}
\noindent The authors thank Fr\'ed\'eric Chyzak, Bruno Salvy, Michael Singer
and anonymous referees for helpful discussions and suggestions.

{\small


\begin{thebibliography}{10}

\smallskip

\bibitem{Abramov2003} S.~A. Abramov.
\newblock When does {Z}eilberger's algorithm succeed?
\newblock {\em Adv. in Appl. Math.}, 30(3):424--441, 2003.

\bibitem{AbramovPetkovsek2001} S.~A. Abramov and M.~Petkov{\v{s}}ek.
\newblock Proof of a conjecture of {W}ilf and {Z}eilberger.
\newblock Preprints Series of the Inst.\ Math, Physics
and Mechanics, 39(748), Ljubljana, 2001.

\bibitem{AbramovPetkovsek2008} S.~A. Abramov and M.~Petkov{\v{s}}ek.
\newblock Dimensions of solution spaces of {$H$}-systems.
\newblock {\em J.\ Symbolic Comput.}, 43(5):377--394, 2008.

\bibitem{AbramovPetkovsek2002a} S.~A. Abramov and M.~Petkov\v{s}ek.
\newblock On the structure of multivariate hypergeometric terms.
\newblock {\em Adv. in Appl. Math.}, 29(3):386--411, 2002.

\bibitem{Bronstein2005} M.~Bronstein, Z.~Li, and M.~Wu.
\newblock Picard--{V}essiot extensions for linear functional systems.
\newblock In {\em Proc.\ of ISSAC '05}, 68--75, ACM. New York, USA, 2005.

\bibitem{CFFLpre2011} S.~Chen, R.~Feng, G.\ Fu and Z.~Li.
\newblock On the structure of compatible rational functions.
\newblock MM-Res. Preprints, 30:~20-38, 2011. (\url{http://www.mmrc.iss.ac.cn/pub/mm30/02-Chen.pdf})

\bibitem{CCFL2010} S.~Chen, F.~Chyzak, R.~Feng, and Z.~Li.
\newblock The existence of telescopers for hyperexponential-hypergeometric functions.
\newblock MM-Res. Preprints, 29:~239-267, 2010. (\url{http://www.mmrc.iss.ac.cn/pub/mm29/13-Chen.pdf})

\bibitem{Chen2005} W.~Y.~C. Chen, Q.-H. Hou, and Y.-P. Mu.
\newblock Applicability of the {$q$}-analogue of {Z}eilberger's algorithm.
\newblock {\em J. Symbolic Comput.}, 39(2):155--170, 2005.

\bibitem{Christopher1999} C.~Christopher.
\newblock Liouvillian first integrals of second order polynomial differential equations.
\newblock {\em Electron. J. Differential Equations}, 49: 1-7 (electronic), 1999.

\bibitem{Feng2010b} R.~Feng, M.~F. Singer, and M.~Wu.
\newblock An algorithm to compute {L}iouvillian solutions of prime order
linear difference-differential equations.
\newblock {\em J. Symbolic Comput.}, 45(3):306--323, 2010.

\bibitem{Gelfand} I.~Gel{'}fand, M.~Graev, and V.~Retakh.
\newblock General hypergeometric systems of equations and series of hypergeometric type.
\newblock {\em Uspekhi Mat. Nauk (Russian), Engl.\ transl.\ in
Russia Math Surveys}, 47(4):3--82, 1992.

\bibitem{Hardouin2008} C.~Hardouin and M.~F.\ Singer.
\newblock Differential {G}alois theory of linear difference equations.
\newblock {\em Math. Ann.}, 342(2):333--377, 2008.

\bibitem{Kauers2005} M.~Kauers.
\newblock {\em {Algorithms for Nonlinear Higher Order Difference Equations}}.
\newblock PhD thesis, RISC-Linz, Linz, Austria, 2005.

\bibitem{LabahnLi2004} G.~Labahn and Z.~Li.
\newblock Hyperexponential solutions of finite-rank ideals in orthogonal {O}re rings.
\newblock In {\em Proc.\ of ISSAC'04}, 213--220. ACM, New York, 2004.

\bibitem{LiWuZheng2007} Z.~Li, M.~Wu, and D.~Zheng.
\newblock Testing linear dependence of hyperexponential elements.
\newblock {\em ACM Commun. Comput. Algebra}, 41(1):3--11, 2007.

\bibitem{Ore1930} O.~Ore.
\newblock Sur la forme des fonctions hyperg{\' e}om{\'e}triques de plusieurs variables.
\newblock {\em J.\ Math.\ Pures Appl.}, 9(4):311--326, 1930.

\bibitem{Payne1997} G.~H.\ Payne.
\newblock {\em Multivariate Hypergeometric Terms}.
\newblock PhD thesis, Penn.\ State Univ., Pennsylvania, USA, 1997.

\bibitem{Sato1990} M.~Sato.
\newblock Theory of prehomogeneous vector spaces (algebraic part)--
the {E}nglish translation of {S}ato's lecture from {S}hintani's note.
\newblock {\em Nagoya Math. J.}, 120:1--34, 1990.

\bibitem{Sing2007} M.F.\ Singer.
\newblock A note on solutions of first-order linear functional equations.
\newblock Manuscript for discussions at the Second NCSU-China
Symbolic Computation Collaboration Workshop, Hangzhou, March, 2007.

\bibitem{AdPutSinger1997} M.~van~der Put and M.F.\ Singer.
\newblock {\em Galois {T}heory of {D}ifference {E}quations}, volume 1666 of {\em Lecture Notes in Mathematics}.
\newblock Springer-Verlag, Berlin, 1997.

\bibitem{Zoladek1998} H.~Zoladek.
\newblock The extended monodromy group and {L}iouvillian first integrals.
\newblock {\em J. Dynam. Control Systems}, 4(1):1--28, 1998.

\end{thebibliography}
\end{document}